\documentclass[11pt]{article}       
\usepackage{amsthm}
\usepackage{amssymb}
\usepackage{float}
\usepackage{lineno}
\usepackage{geometry}
\usepackage{setspace}\onehalfspace
\newtheorem{theorem}{Theorem}[section]
\newtheorem{claim}{Claim}

\begin{document}

\title{The Geodetic Hull Number is Hard for Chordal Graphs}

\date{}

\author{St\'ephane Bessy$^1$ \and
Mitre C. Dourado$^2$ \and
Lucia D. Penso$^3$ \and 
Dieter Rautenbach$^3$}

\maketitle

\begin{center}
{\small
$^1$ 
Laboratoire d'Informatique, de Robotique et de Micro\'{e}lectronique de Montpellier\\
Universit\'e de Montpellier, Montpellier, France, stephane.bessy@lirmm.fr\\[3mm]
$^2$ Departamento de Ci\^{e}ncia da Computa\c{c}\~{a}o, Instituto de Matem\'{a}tica\\
Universidade Federal do Rio de Janeiro, Rio de Janeiro, Brazil, mitre@dcc.ufrj.br\\[3mm]
$^3$ Institute of Optimization and Operations Research\\ 
Ulm University, Ulm, Germany, $\{$lucia.penso,dieter.rautenbach$\}$@uni-ulm.de}
\end{center}

\begin{abstract}
We show the hardness of the geodetic hull number for chordal graphs.
\end{abstract}

{\small
\begin{tabular}{lp{12.5cm}}
\textbf{Keywords:} & Geodetic convexity; shortest path; hull number; chordal graphs
\end{tabular}
}

\section{Introduction} \label{sec:int}

One of the most well studied convexity notions for graphs is the {\it shortest path convexity} or {\it geodetic convexity},
where a set $X$ of vertices of a graph $G$ is considered {\it convex}
if no vertex outside of $S$ lies on a shortest path between two vertices inside of $S$.
Defining the {\it convex hull} of a set $S$ of vertices as the smallest convex set containing $S$,
a natural parameter of $G$ is its {\it hull number} $h(G)$ \cite{es},
which is the minimum order of a set of vertices whose convex hull is the entire vertex set of $G$.
The hull number is NP-hard 
for bipartite graphs \cite{acgnss}, 
partial cubes \cite{ak}, and 
$P_9$-free graphs \cite{dpr},
but it can be computed in polynomial time for 
cographs \cite{dgkps}, 
$(q,q-4)$-graphs \cite{acgnss}, 
$\{ {\rm paw},P_5\}$-free graphs \cite{amssw,dpr}, and
distance-hereditary graphs \cite{kn}.
Bounds on the hull number are given in \cite{acgnss,dprs,es}.

In \cite{kn} Kant\'e and Nourine present a polynomial time algorithm for the computation of the hull number of chordal graphs.
Unfortunately, their correctness proof contains a gap 
described in detail at the end of the present paper.
As our main result we show that computing the hull number of a chordal graph is NP-hard,
which most likely rules out the existence of a polynomial time algorithm.

Before we proceed to our results, we collect some notation and terminology.
We consider finite, simple, and undirected graphs.
A graph $G$ has vertex set $V(G)$ and edge set $E(G)$.
A graph $G$ is {\it chordal} if it does not contain an induced cycle of order at least $4$.
A {\it clique} in $G$ is the vertex set of a complete subgraph of $G$.
A vertex of a graph $G$ is {\it simplicial} in $G$ if its neighborhood is a clique.
The {\it distance} ${\rm dist}_G(u,v)$ between two vertices $u$ and $v$ in $G$ 
is the minimum number of edges of a path in $G$ between $u$ and $v$. 
The {\it diameter} ${\rm diam}(G)$ of $G$ is the maximum distance between any two vertices of $G$.
The {\it eccentricity} $e_G(u)$ of a vertex $u$ of $G$ is the maximum distance between $u$ and any other vertex of $G$. 
For a positive integer $k$, let $[k]$ be the set of the positive integers at most $k$.

Let $G$ be a graph, and let $S$ be a set of vertices of $G$.
The {\it interval} $I_G(S)$ of $S$ in $G$
is the set of all vertices of $G$ that lie on shortest paths in $G$ between vertices from $S$. 
Note that $S\subseteq I_G(S)$,
and that $S$ is {\it convex} in $G$ if $I_G(S)=S$. 
The set $S$ is {\it concave} in $G$ 
if $V(G) \setminus S$ is convex. 
Note that $S$ is concave if and only if 
$S \cap I_G(\{v,w\}) = \emptyset$ 
for every two vertices $v$ and $w$ in $V(G)\setminus S$. 
The {\it hull} $H_G(S)$ of $S$ in $G$,
defined as the smallest convex set in $G$ that contains $S$,
equals the intersection of all convex sets that contain $S$.
The set $S$ is a {\it hull set} if $H_G(S)=V(G)$,
and the {\it hull number} $h(G)$ of $G$ \cite{dpr,es} 
is the smallest order of a hull set of $G$.

\section{Result}

We immediately proceed to our main result.

\begin{theorem}\label{theorem1}
For a given chordal graph $G$, and a given integer $k$, it is NP-complete to decide whether the hull number $h(G)$ of $G$ is at most $k$.
\end{theorem}
\begin{proof}
Since the hull of a set of vertices of $G$ can be computed in polynomial time, the considered decision problem belongs to NP. 
In order to prove NP-completeness, we describe a polynomial reduction from a restricted version of {\sc Satisfiability}.
Therefore,  
let ${\cal C}$ be an instance of {\sc Satisfiability} 
consisting of $m$ clauses $C_1,\ldots,C_m$ 
over $n$ boolean variables $x_1,\ldots,x_n$
such that 
every clause in ${\cal C}$ contains at most three literals, and, 
for every variable $x_i$, 
there are exactly two clauses in ${\cal C}$,
say $C_{j_i^{(1)}}$ and  $C_{j_i^{(2)}}$, 
that contain the literal $x_i$,
and exactly one clause in ${\cal C}$, 
say $C_{j_i^{(3)}}$, 
that contains the literal $\bar{x}_i$, 
and these three clauses are distinct.
Using a polynomial reduction from [LO1] \cite{gj},
it has been shown in \cite{dpr} 
that {\sc Satisfiability} restricted to such instances 
is still NP-complete.

\begin{figure}[H]
\label{fig1}
\begin{center}
\unitlength 1.65mm 
\linethickness{0.4pt}
\ifx\plotpoint\undefined\newsavebox{\plotpoint}\fi 
\begin{picture}(87,39)(0,0)
\put(8,5){\circle*{1}}
\put(8,15){\circle*{1}}
\put(18,15){\circle*{1}}
\put(23,5){\circle*{1}}
\put(28,15){\circle*{1}}
\put(38,5){\circle*{1}}
\put(23,22){\circle*{1}}
\put(38,15){\circle*{1}}
\put(23,32){\circle*{1}}
\put(44,32){\circle*{1}}
\put(53,5){\circle*{1}}
\put(63,5){\circle*{1}}
\put(53,15){\circle*{1}}
\put(58,32){\circle*{1}}
\put(63,15){\circle*{1}}
\put(8,15){\line(1,0){10}}
\put(18,15){\line(1,-2){5}}
\put(23,5){\line(-3,2){15}}
\put(28,15){\line(-1,-2){5}}
\put(23,5){\line(3,2){15}}
\put(53,15){\line(1,-1){10}}
\put(63,5){\line(0,1){10}}
\put(63,15){\line(-1,0){10}}
\put(23,32){\line(0,-1){10}}
\multiput(23,22)(-.0437317784,-.0204081633){343}{\line(-1,0){.0437317784}}
\multiput(23,22)(.0437317784,-.0204081633){343}{\line(1,0){.0437317784}}
\put(38,15){\line(0,-1){10}}
\put(53,15){\line(0,-1){10}}
\put(36,4){\oval(62,8)[]}
\put(38,1){\makebox(0,0)[cc]{}}
\put(23,3){\makebox(0,0)[cc]{$y_i$}}
\put(63,3){\makebox(0,0)[cc]{$\bar{y}_i$}}
\put(46,30){\makebox(0,0)[cc]{$z_i$}}
\put(45,5){\makebox(0,0)[cc]{$B$}}
\put(44,37){\makebox(0,0)[cc]{$Z \setminus \{z_i\}$}}
\put(20.5,17){\makebox(0,0)[cc]{$x'^1_i$}}
\put(26.5,17){\makebox(0,0)[cc]{$x'^2_i$}}
\put(21,24){\makebox(0,0)[cc]{$x'_i$}}
\put(8,3){\makebox(0,0)[cc]{$c_j$}}
\put(38,3){\makebox(0,0)[cc]{$c_k$}}
\put(53,3){\makebox(0,0)[cc]{$c_{\ell}$}}
\qbezier(23,32)(53.5,21.5)(38,5)
\qbezier(23,32)(-7.5,21.5)(8,5)
\qbezier(23,22)(51.5,22.5)(38,5)
\qbezier(23,22)(-5.5,22.5)(8,5)
\put(37.5,17.5){\makebox(0,0)[cc]{$x^2_i$}}
\put(8,15){\line(0,-1){10}}
\put(8,5){\line(0,1){0}}
\put(23,22){\line(0,-1){17}}
\put(23,5){\line(0,1){0}}
\put(8,17.5){\makebox(0,0)[cc]{$x^1_i$}}
\put(21,29.5){\makebox(0,0)[cc]{$x_i$}}
\put(51,17){\makebox(0,0)[cc]{${\bar x}'_i$}}
\put(64,17){\makebox(0,0)[cc]{${\bar x}''_i$}}
\put(60,29){\makebox(0,0)[cc]{${\bar x}_i$}}
\qbezier(23,32)(59,11.5)(23,5)
\qbezier(51,37)(87,35.5)(67,4)
\put(28,15){\line(1,0){10}}
\put(44,35){\line(0,-1){27}}
\put(44,37){\oval(14,4)[]}
\put(23,32){\line(1,0){35}}
\multiput(58,32)(-.0204081633,-.0693877551){245}{\line(0,-1){.0693877551}}
\qbezier(58,32)(36.5,25.5)(53,5)
\qbezier(58,32)(79.5,25.5)(63,5)
\end{picture}
\end{center}
\caption{The vertices and edge added for the variable $x_i$, where 
$j_i^{(1)}=j$,
$j_i^{(2)}=k$, and
$j_i^{(3)}=\ell$.}
\end{figure}
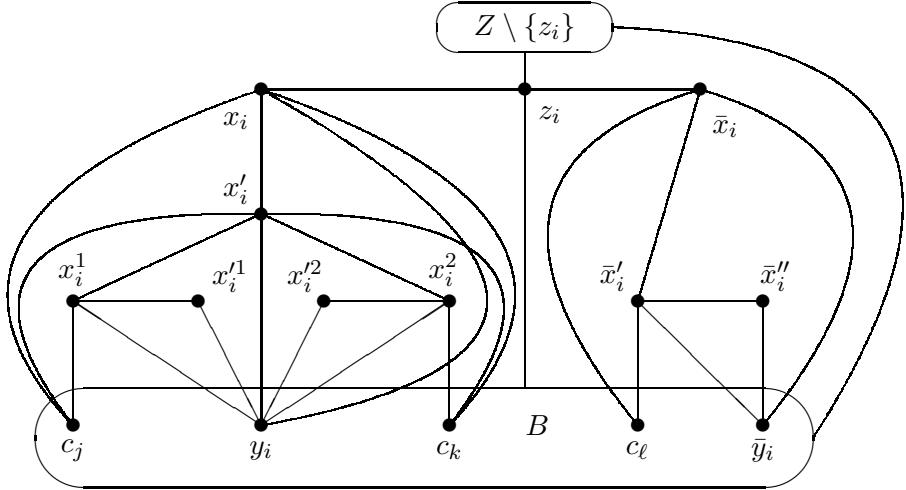
\noindent Let the graph $G$ be constructed as follows starting with the empty graph:
\begin{itemize}
\item For every $j\in [m]$, add a vertex $c_j$. 
\item For every $i\in [n]$, add three $y_i,\bar{y}_i,$ and $z_i$. 
\item Add edges such that $B \cup Z$ is a clique, where 
\begin{eqnarray*}
B& = &\{c_j : j \in [m]\} \cup \{y_i : i \in [n] \}\cup \{\bar{y}_i : i \in [n] \}
\mbox{ and }\\
Z &=& \{z_i : i \in [n]\}\mbox{, and }
\end{eqnarray*}
\item For every $i\in [n]$, add $9$ vertices and $25$ edges to obtain the subgraph indicated in Figure \ref{fig1}.
\end{itemize}
Note that ${\rm dist}_G(x_i,\bar{x}_i')={\rm dist}_G(\bar{x}_i,x'^1_i)=3$ for every $i$ in $[n]$.
Since every vertex of $G$ has a neighbor in the clique $B \cup Z$, 
the diameter of $G$ is $3$.
Furthermore, since no vertex is universal, 
all vertices in $B \cup Z$ have eccentricity $2$.

Let $k = 4n$. 

Note that the order of $G$ is $12n+m$. 

It remains to show that $G$ is chordal, 
and that $\cal C$ is satisfiable if and only if $h(G)\leq k$.

\medskip

\noindent In order to show that $G$ is chordal,
we indicate a {\it perfect elimination ordering}, 
which is a linear ordering $v_1,\ldots,v_{12n+m}$ of its vertices such that $v_i$ is simplicial in $G-\{v_1, \ldots v_{i-1}\}$ for every $i$ in $[12n+m]$. 
Such an ordering is obtained by 
\begin{itemize}
\item starting with the vertices $x'^1_i, x'^2_i,$ and $\bar{x}''_i$ for all $i \in [n]$ (in any order),
\item continuing with the vertices $x^1_i, x^2_i,$ and $\bar{x}'_i$ for all $i \in [n]$,
\item continuing with the vertices $x'_i$ for all $i \in [n]$,
\item continuing with the vertices $x_i$ and $\bar{x}_i$ for all $i \in [n]$, and
\item ending with the vertices in the clique $B \cup Z$.
\end{itemize}
Now, let $\cal S$ be a satisfying truth assignment for $\cal C$.

Let 
$$S=\bigcup_{i\in [n]}\left\{x'^1_i, x'^2_i,\bar{x}''_i\right\}\,\,
\cup\,\,\bigcup_{i\in [n]:\,\,x_i\,\,true\,\,in\,\,{\cal S}}\left\{ x_i\right\}\,\,
\cup\,\,\bigcup_{i\in [n]:\,\,x_i\,\,false\,\,in\,\,{\cal S}}\left\{ \bar{x}_i\right\}.$$
Clearly, $|S|=k=4n$.
For every $i$ in $[n]$,
we have 
$\{z_i,\bar{y}_i\} \subseteq I_G(\{ x_i,\bar{x}''_i\})$,
$\{z_i,y_i\} \subseteq I_G(\{ \bar{x}_i,x'^1_i\})$,
$y_i \in I_G(\{ \bar{y}_i,x'^1_i\})$, and 
$\bar{y}_i \in I_G(\{ y_i,\bar{x}''_i\})$,
which implies $\{z_i,y_i,\bar{y}_i\} \subseteq H_G(S)$.
Since ${\cal S}$ is a satisfying truth assignment, 
for every $j$ in $[m]$, 
there is a neighbor, say $v$, of $c_j$ in 
$$\bigcup_{i\in [n]:\,\,x_i\,\,true\,\,in\,\,{\cal S}}\left\{ x_i\right\}\,\,\cup\,\,\bigcup_{i\in [n]:\,\,x_i\,\,false\,\,in\,\,{\cal S}}\left\{ \bar{x}_i\right\}.$$ 
If $v \in \bigcup\limits_{i\in [n]:\,\,x_i\,\,true\,\,in\,\,{\cal S}}\left\{ x_i\right\}$, then $c_j \in I_G(\{ v,\bar{x}''_i\})$, otherwise $c_j \in I_G(\{ v,x'^1_i\})$. 
Hence, $B \cup Z\subseteq H_G(S)$.

Now, for some $i$ in $[n]$, 
let $c_j$, $c_k$, and $c_\ell$ 
be the neighbors in 
$B \setminus \{ y_i,\bar{y}_i\}$ of $x^1_i$, $x^2_i$, and $\bar{x}'_i$, respectively, 
similarly as in Figure \ref{fig1}. 
We have
$x^1_i \in I_G(\{ x'^1_i,c_j\})$, 
$x^2_i \in I_G(\{ x'^2_i,c_k\})$, 
$x_i' \in I_G(\{ x^1_i,x^2_i\})$, 
$\bar{x}'_i \in I_G(\{ \bar{x}''_i,c_\ell\})$,
$x_i \in I_G(\{ x_i',z_i\})$, and 
$\bar{x}_i\in I_G(\{ \bar{x}'_i,z_i\})$. 

Altogether, we obtain that $S$ is a hull set of $G$ of order $4n$.

\medskip

\noindent Finally, let $S$ be a hull set of $G$ of order at most $4n$. 

\begin{claim}\label{claim1}
For every $i \in [n]$, the set $\{x_i,z_i,\bar{x}_i\}$ is concave.
\end{claim} 
\noindent {\it Proof of Claim \ref{claim1}:}
For a contradiction, suppose that some vertex in $S'=\{x_i,z_i,\bar{x}_i\}$
lies on a shortest path $P$ in $G$ between two vertices $v$ and $w$ in $V(G)\setminus S'$.
Since the diameter of $G$ is $3$, 
the path $P$ contains at most $2$ vertices of $S'$. 
Since the neighbors outside of $S'$ of each vertex in $S'$ form a clique, 
the path $P$ contains exactly $2$ adjacent vertices of $S'$,
that is, either $P = vx_iz_iw$ or $P = v\bar{x}_iz_iw$.
In both cases, the vertex $w$ has eccentricity at least $3$.
However, every neighbor $w$ of $z_i$ outside $S'$ belongs to $B \cup Z$, 
and thus, has eccentricity $2$, a contradiction.
$\Box$

\begin{claim}\label{claim2}
For every $j \in [m]$, the set
$$V_j=\{ c_j\}
\cup\,\,\bigcup_{i\in [n]:j=j_i^{(1)}}\left\{ x_i,x'_i,x_i^{1}\right\}\,\,
\cup\,\,\bigcup_{i\in [n]:j=j_i^{(2)}}\left\{ x_i,x'_i,x_i^{2}\right\}\,\,
\cup\,\,\bigcup_{i\in [n]:j=j_i^{(3)}}\left\{ \bar{x}_i,\bar{x}'_i\right\}$$
is concave.
\end{claim}
\noindent {\it Proof of Claim \ref{claim2}:}
First, suppose that $C_j$ contains the positive literal $x_i$.
By symmetry, 
we may assume that $j=j_i^{(1)}$ and $j_i^{(2)}=k$
for some $k$ in $[m]\setminus \{ j\}$.

First, suppose that some shortest path $P$ 
between two vertices $v$ and $w$ in $\bar{V}_j=V(G)\setminus V_j$
contains $x_i$.
Choosing $P$ of minimum length, 
it follows that $v$ and $w$ are the only vertices of $P$ in $\bar{V}_j$.
Since the diameter of $G$ is $3$,
the length of $P$ is at most $3$, 
and we may assume that $v$ is a neighbor of $x_i$,
which implies $v \in \{z_i, c_k, y_i\}$. 
Since $\{z_i, c_k, y_i\}$ is a clique, 
the vertex $w$ is not a neighbor of $x_i$, 
and $P$ contains exactly one vertex $u$ of $V_j$ different of $x_i$,
which implies $P=vx_iuw$ and $u \in \{x'_i, c_j\}$. 
Suppose that $u = x'_i$.
This implies $w \in \{x^2_i, c_k,y_i\}$. 
Since $c_k,y_i \in N_G(x_i)$, 
we obtain $w = x^2_i$ and $v = z_i$. 
However, ${\rm dist}_G(z_i,x^2_i) = 2$,
which is a contradiction.
Hence, $u = c_j$ and $w \in B \cup Z$. 
However, every vertex in $B \cup Z$ has eccentricity $2$,
which is a contradiction. 
Hence, no shortest path between two vertices in $\bar{V}_j$ contains $x_i$.

Next, suppose that some shortest path $P$ 
between two vertices $v$ and $w$ in $\bar{V}_j$
contains $x_i'$.
Similarly as above, 
we may assume that 
$v$ and $w$ are the only vertices of $P$ in $\bar{V}_j$,
the length of $P$ is at most $3$, 
and $v$ is a neighbor of $x_i'$,
which implies $v\in \{x^2_i, y_i, c_k\}$.
Since $\{x^2_i, y_i, c_k\}$ is a clique, 
the path $P$ contains exactly one vertex $u$ of $V_j$ different of $x_i'$,
which implies $P=vx'_iuw$ and $u \in \{x_i^1, c_j\}$,
where we use that $P$ does not contain $x_i$.
Suppose that $u = x_i^1$.
This implies $w \in \{ x'^1_i,y_i\}$. 
Since $y_i \in N_G(x'_i)$, 
we obtain $w =x'^1_i$
and $v =x^2_i$.
However, ${\rm dist}_G(x^2_i,x'^1_i) = 2$,
which is a contradiction.
Hence, $u = c_j$ and $w \in B \cup Z$. 
However, every vertex in $B \cup Z$ has eccentricity $2$,
which is a contradiction. 
Hence, no shortest path between two vertices in $\bar{V}_j$ contains $x'_i$.

Next, suppose that some shortest path $P$ 
between two vertices $v$ and $w$ in $\bar{V}_j$
contains $x_i^1$.
Similarly as above, 
we may assume that 
$v$ and $w$ are the only vertices of $P$ in $\bar{V}_j$,
the length of $P$ is at most $3$, 
and $v$ is a neighbor of $x_i^1$,
which implies $v\in \{ x'^1_i,y_i\}$.
Since $\{ x'^1_i,y_i\}$ is a clique, 
the path $P$ contains exactly one vertex $u$ of $V_j$ different of $x_i^1$,
which implies $P=vx_i^1c_jw$ and $w\in B\cup Z$,
where we use that $P$ does not contain $x'_i$. 
However, every vertex in $B \cup Z$ has eccentricity $2$,
which is a contradiction. 
Hence, no shortest path between two vertices in $\bar{V}_j$ contains $x_i^1$.

\medskip

\noindent Next, suppose that $C_j$ contains the negative literal $\bar{x}_i$,
that is, $j=j_i^{(3)}$.

First, suppose that some shortest path $P$ 
between two vertices $v$ and $w$ in $\bar{V}_j$
contains $\bar{x}_i$.
Similarly as above, 
we may assume that 
$v$ and $w$ are the only vertices of $P$ in $\bar{V}_j$,
the length of $P$ is at most $3$, 
and $v$ is a neighbor of $\bar{x}_i$,
which implies $v\in \{ z_i,\bar{y}_i\}$.
Since $\{ z_i,\bar{y}_i\}$ is a clique, 
the vertex $w$ is not a neighbor of $\bar{x}_i$, 
and $P$ contains exactly one vertex $u$ of $V_j$ different of $\bar{x}_i$,
which implies $P=v\bar{x}_iuw$ and $u \in \{\bar{x}'_i, c_j\}$. 
Suppose that $u = \bar{x}'_i$.
This implies $w \in \{ \bar{x}''_i,\bar{y}_i\}$. 
Since $\bar{y}_i\in N_G(\bar{x}_i)$, 
we obtain $v = z_i$ and $w =\bar{x}''_i$. 
However, ${\rm dist}_G(z_i,\bar{x}''_i) = 2$,
which is a contradiction.
Hence, $u = c_j$ and $w \in B \cup Z$. 
However, every vertex in $B \cup Z$ has eccentricity $2$,
which is a contradiction. 
Hence, no shortest path between two vertices in $\bar{V}_j$ contains $\bar{x}_i$.

Next, suppose that some shortest path $P$ 
between two vertices $v$ and $w$ in $\bar{V}_j$
contains $\bar{x}_i'$.
Similarly as above, 
we may assume that 
$v$ and $w$ are the only vertices of $P$ in $\bar{V}_j$,
the length of $P$ is at most $3$, 
and $v$ is a neighbor of $\bar{x}_i'$,
which implies $v\in \{ \bar{x}_i'',\bar{y}_i\}$.
Since $\{\bar{x}_i'',\bar{y}_i\}$ is a clique, 
the path $P$ contains exactly one vertex $u$ of $V_j$ different of $\bar{x}_i'$,
which implies $P=v\bar{x}_i'c_jw$ and $w\in B\cup Z$,
where we use that $P$ does not contain $\bar{x}_i$. 
However, every vertex in $B \cup Z$ has eccentricity $2$,
which is a contradiction. 
Hence, no shortest path between two vertices in $\bar{V}_j$ contains $\bar{x}_i'$.

Finally, since the neighbors of $c_j$ outside of $V_j$ form a clique,
no shortest path between two vertices in $\bar{V}_j$ contains $c_j$,
which completes the proof of the claim.
$\Box$

\medskip

\noindent Note that all $3n$ simplicial vertices in 
$\bigcup\limits_{i\in [n]}\left\{x'^1_i, x'^2_i,\bar{x}''_i\right\}$
belong to $S$.

Since $S$ contains at most $n$ non-simplicial vertices,
Claim \ref{claim1} implies that, for every $i$ in $[n]$,
the set $S$ contains exactly one of the three vertices in $\{x_i,z_i,\bar{x}_i\}$,
and that these are the only non-simplicial vertices in $S$.
Now, Claim \ref{claim2} implies that, for every $j$ in $[m]$,
there is some $i\in [n]$ such that
\begin{itemize}
\item either $C_j$ contains the literal $x_i$ and the vertex $x_i$ belongs to $S$
\item or $C_j$ contains the literal $\bar{x}_i$ and the vertex $\bar{x}_i$ belongs to $S$.
\end{itemize}
Therefore, setting the variable $x_i$ to true if and only if the vertex $x_i$ belongs to $S$
yields a satisfying truth assignment ${\cal S}$ for ${\cal C}$,
which completes the proof. 
\end{proof}
\noindent As pointed out in the introduction, the correctness proof in \cite{kn} contains a gap.
In lines 14 and 15 on page 322 of \cite{kn} it says 
\begin{quote}
``{\it At iteration $i+1$, the vertex $x_{i+1}$ is a simplicial vertex in $G_{i+1}$. 
We first claim that there exists no functional dependency of the form $zt\to x_{i+1}$ in $\Sigma$.}''
\end{quote}
Consider applying the algorithm from \cite{kn} to the graph in Figure \ref{figex}.
In iteration 1, it would decide to add $x_1$ to $K$.
In iteration 2, it would decide not to add $x_2$ to $K$, because of $t\to x_2$. 
Furthermore, because of $t\to x_2$ and $z,x_2\to x_3$, it would replace $z,x_2\to x_3$ within $\Sigma$ with $z,t\to x_3$.
Therefore, in iteration 3, $\Sigma$ would actually contain $z,t\to x_3$, contrary to the claim cited above.

\begin{figure}[H]
\label{figex}
\begin{center}
\unitlength 1.2mm 
\linethickness{0.4pt}
\ifx\plotpoint\undefined\newsavebox{\plotpoint}\fi 
\begin{picture}(66,15)(0,0)
\put(5,5){\circle*{1.5}}
\put(20,5){\circle*{1.5}}
\put(35,5){\circle*{1.5}}
\put(50,5){\circle*{1.5}}
\put(65,5){\circle*{1.5}}
\put(5,0){\makebox(0,0)[cc]{$x_1$}}
\put(20,0){\makebox(0,0)[cc]{$x_2$}}
\put(35,0){\makebox(0,0)[cc]{$x_3$}}
\put(50,0){\makebox(0,0)[cc]{$t$}}
\put(65,0){\makebox(0,0)[cc]{$z$}}
\put(65,5){\line(-1,0){60}}
\qbezier(20,5)(35,15)(50,5)
\qbezier(65,5)(50,15)(35,5)
\end{picture}
\end{center}
\caption{A small chordal graph.}
\end{figure}
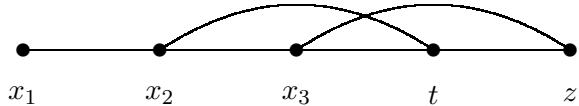

\end{document}